\documentclass[12pt]{amsart}

\usepackage{amsmath,amsthm,amssymb,amsfonts,verbatim}
\usepackage[colorlinks,
            linkcolor=blue,
            anchorcolor=blue,
            citecolor=blue
            ]{hyperref}
\usepackage[hmargin=1.2in,vmargin=1.2in]{geometry}

\title[Codes from algebraic functions fields]{A new construction of nonlinear codes via algebraic function fields}
\author{Shu Liu}\address{Natl Key Lab Sci and Technol Commun, University of Electronic Science and Technology of China, Chengdu, China.} \email{shuliu@uestc.edu.cn}
\author{Liming Ma}\address{School of Mathematical Sciences, University of Science and Technology of China, Hefei 230026, China
}\email{lmma20@ustc.edu.cn}
\author{ Ting-Yi~Wu}\address{Theory Lab, Central Research Institute, 2012 Labs, Huawei Technologies Co. Ltd.,
Hong Kong SAR, China
}\email{wu.ting.yi@huawei.com}
\author{Chaoping Xing} \address{School of Electronics, Information and Electric Engineering, Shanghai Jiao Tong University,
China 200240}\email{xingcp@sjtu.edu.cn}
\date{}

\newtheorem{lemma}{Lemma}[section]
\newtheorem{theorem}[lemma]{Theorem}
\newtheorem{cor}[lemma]{Corollary}

\newtheorem{prop}[lemma]{Proposition}
\newtheorem{ex}[lemma]{Example}

\theoremstyle{remark}
\newtheorem{rmk}{Remark}

\renewcommand{\epsilon}{\varepsilon}
\renewcommand{\le}{\leqslant}
\renewcommand{\ge}{\geqslant}

\newcommand{\vnote}[1]{}

\def\ZZ{\mathbb{Z}}
\def\PP{\mathbb{P}}

\def\F{\mathbb{F}}
\def\Z{\mathbb{Z}}

\def \mL {\mathcal{L}}

\def \Xi {{X^{[i]}}}

\newcommand{\Ga}{\alpha}

\def \bx {{\bf x}}

\def \by {{\bf y}}

\def \bu {{\bf u}}
\def \bv {{\bf v}}

\def\supp {{\rm supp }}

\begin{document}
\maketitle

\begin{abstract}
In coding theory, constructing codes with good parameters is one of the most important and fundamental problems.
Though a great many of good codes have been produced, most of them are defined over alphabets of sizes equal to prime powers.
In this paper, we provide a new explicit construction of $(q+1)$-ary nonlinear codes via algebraic function fields, where $q$ is a prime power.
Our codes are constructed by evaluations of rational functions at all rational places of the algebraic function field.
Compared with algebraic geometry codes, the main difference is that we allow rational functions to be evaluated at pole places.
After evaluating rational functions from a union of Riemann-Roch spaces, we obtain a family of nonlinear codes over the alphabet $\mathbb{F}_{q}\cup \{\infty\}$.
It turns out that our codes have better parameters than those obtained from MDS codes or good algebraic geometry codes via code alphabet extension and restriction.
\end{abstract}

\section{Introduction}

In coding theory, constructing codes with good parameters is one of the most important and fundamental problems.
For a $q$-ary code of length $n$, size $M$ and minimum distance $d$, we denote it by an $(n,M,d)$-code. The size is a measure of its efficiency and the minimum distance represents its error-correcting capability. Hence, people hope that both the size $M$ and minimum distance $d$ are as large as possible.
However, there is a trade-off between the size and the minimum distance of the code.
One of the well-known upper bounds is the Singleton bound which says that $M\le q^{n-d+1}$. A linear code achieving this bound is called a maximum distance separable (MDS) code.

Many efforts have been devoted to various constructions of good codes. Linear codes have received a lot of attention, such as Reed-Solomon codes, BCH codes, cyclic codes and so on, since they have good algebraic structures and many practical advantages.
However, there are some examples showing that linear codes do not exist for some parameters that nonlinear codes can have.
For example, there are no binary linear codes with parameters $[16,8,6]$. On the other hand, the Nordstorm-Robinson code \cite{LX04} is a binary nonlinear code with parameters $(16,2^8,6)$.
Furthermore, the Nordstrom-Robinson code can be viewed as an image under the Gray map of some algebraic geometry code over $\ZZ/4\ZZ$ in \cite{W97}.
Therefore, it is also of interest to provide explicit constructions of nonlinear codes.
Though a large number of nonlinear codes have been constructed, most of them are $q$-ary codes where $q$ is a prime power. 
Less is known for constructions of $q$-ary codes, where $q$ is not a prime power. Some nonlinear codes over  $\ZZ_4$, $\ZZ_6$, $\ZZ_{10}$ or $\ZZ_{12}$ were given with certain properties \cite{GH05,H00, HM11}.

In \cite{JMX21}, an explicit construction of $(q+1)$-ary $(q+1, q^{2m+1}+q^{2m}-2q^m+2, q+1-2m)$ nonlinear codes with $q$ being a prime power was presented. Such codes have better parameters than those obtained from MDS codes via code alphabet restriction and extension. Another advantage of these codes is that they can be efficiently decoded. Due to rich structures of algebraic function fields over finite fields, various techniques have been employed to construct good codes from algebraic function fields  \cite{BTV17, JMX20,LMX19, MX20, MX21, X11}.

In this paper, we generalizes the construction of nonlinear codes via rational function fields given in \cite{JMX21} to algebraic function fields. Our nonlinear codes are constructed by evaluations of rational functions at all rational places of algebraic function fields.
Compared with algebraic geometry codes, the main difference is that we allow rational functions to be evaluated at pole places.
After evaluating rational functions from a union of Riemann-Roch spaces, we construct a family of good nonlinear codes over the alphabet $\mathbb{F}_{q}\cup \{\infty\}$. Note that code sizes in  \cite{JMX21} are exactly calculated due to the nature of rational function fields, while  lower bounds on code sizes in this paper are provided.

This paper is organized as follows. In Section \ref{sec:2}, we introduce the basic facts on algebraic function fields, Riemann-Roch spaces, Zeta functions, codes and algebraic geometry codes. In Section \ref{sec:3}, we give an explicit construction of $(q+1)$-ary nonlinear codes from algebraic function fields over the finite field $\F_q$. In particular, we focus on the constructions of nonlinear codes via elliptic curves in Section \ref{sec:4} and maximal function fields in Section \ref{sec:5}, respectively.

\section{Preliminaries}\label{sec:2}
In this section, we present preliminaries on the definitions of algebraic function fields, Riemann-Roch spaces, Zeta functions, Codes and algebraic geometry codes.

\subsection{Algebraic function fields}
Let $q$ be a prime power, let $\F_q$ be the finite field with $q$ elements and let $F/\F_q$ be an algebraic function field with the full constant field $\F_q$.
The set of all places of $F$ is denoted by $\PP_F$.
Let $P\in \PP_F$ be a place of $F$ and let $\mathcal{O}_P$ be its corresponding valuation ring.
The degree of $P$ is defined as the degree of field extension $[\mathcal{O}_P/P:\F_q]$.
A place of $F/\F_q$ with degree one is called rational.
For any rational place $P$ and $f\in \mathcal{O}_P$,  we define $f(P)\in \mathcal{O}_P/P\cong \F_q$ to be the residue class of $f$ modulo $P$; otherwise $f(P)=\infty$ for any $f\in F\setminus \mathcal{O}_P$.

A divisor $G$ of $F$ is a formal sum $G=\sum_{P\in \PP_F} n_PP$ with only finitely many nonzero coefficients $n_P\in \mathbb{Z}$.
The support of $G$ is defined as $\supp(G)=\{P\in \PP_F: n_P\neq 0\}$.
If all coefficients of $G$ are non-negative, then the divisor $G$ is called effective.
Let $\nu_P$ be the normalized discrete valuation of $P$. For any nonzero element $f\in F$, the zero divisor of $f$ is defined by $(f)_0=\sum_{P\in \PP_F, \nu_{P}(f)>0} \nu_P(f)P,$
and the pole divisor of $f$ is defined by $(f)_\infty=\sum_{P\in \PP_F, \nu_P(f)<0} -\nu_P(f)P.$
The principal divisor of $f$ is given by $$(f):=(f)_0-(f)_\infty=\sum_{P\in \PP_F} \nu_P(f)P.$$
For two divisors $G=\sum_{P\in \PP_F} n_PP$ and $D=\sum_{P\in \PP_F} m_PP$, we define the union and intersection of $G$ and $D$ respectively as follows
\[G\vee D:=\sum_{P\in\PP_F} \max\{n_P,m_P\}P,\qquad G\wedge D:=\sum_{P\in\PP_F} \min\{n_P, m_P\}P.\]
It is clear that $G\wedge D +  G\vee D= G+D.$

\subsection{Riemann-Roch spaces}
Let $F/\F_q$ be an algebraic function field with genus $g$.
For a divisor $G$ of $F/\F_q$, the Riemann-Roch space of $G$ is defined by
\[\mL(G):=\{u\in F^*:\; (u)+G\ge 0\}\cup\{0\}.\]
From the Riemann-Roch theorem \cite[Theorem 1.5.17]{St09}, $\mL(G)$ is a $\F_q$-vector space of dimension $\ell(G)\ge \deg(G)-g+1$. Moreover, the equality holds true if $\deg(G)\ge 2g-1$.
For any two divisors $G$ and $H$, it is straightforward to verify that
\begin{equation*}
\mL(G) \cap \mL(H) = \mL(G\wedge H) \ \makebox{and} \ \mL(G)+\mL(H) \subseteq \mL(G\vee H).
\end{equation*}

\begin{lemma}\label{lem:2.1}
Let $f_1,f_2$ be two nonzero functions in $F$ with pole divisors $(f_i)_\infty=G_i$ for $i=1,2$.
If $f_1(P)=f_2(P)\in \mathbb{F}_q\cup \{\infty\}$ for some rational place $P\in  \mathbb{P}_F$, then we have $f_1-f_2\in\mL(G_1+ G_2-P)$.
\end{lemma}
\begin{proof}
{\bf Case 1: } If $f_1(P)=f_2(P)\in\F_q$, then
we have $(f_1-f_2)(P)=f_1(P)-f_2(P)=0$, i.e., $P$ is a zero of $f_1-f_2$. Let $G=(f_1-f_2)_\infty$. Thus, we have $f_1-f_2\in\mL(G-P)$.
Since $f_1-f_2\in \mL(G_1)+\mL(G_2)\subseteq \mL(G_1\vee G_2)$, it follows that  $f_1-f_2\in\mL(G_1\vee G_2-P)\subseteq \mL(G_1+G_2-P)$.

{\bf Case 2: } If $f_1(P)=f_2(P)=\infty$, then we have $P \in\supp(G_1)\cap\supp( G_2)=\supp(G_1\wedge G_2)$.
From the equation $G_1\vee G_2=G_1+G_2-G_1\wedge G_2$, we have $G_1\vee G_2\le G_1+G_2-P$. Since $f_1-f_2\in \mL(G_1)+\mL(G_2)\subseteq \mL(G_1\vee G_2)$, it follows that  $f_1-f_2\in\mL(G_1+ G_2-P)$.
\end{proof}

\subsection{Zeta functions}
Let $F/\F_q$ be an algebraic function field with genus $g$. Let $A_i$ be the number of all effective divisors of $F/\F_q$ of degree $i\ge 0$.
The Zeta function of $F/\F_q$ is defined as the power series
$Z(t):=\sum_{i=0}^\infty A_it^i\in \mathbb{C}[[t]].$
From \cite[Theorem 5.1.15]{St09}, the Zeta function $Z(t)$ can be written as a rational function
$$Z(t)=\frac{L(t)}{(1-t)(1-qt)},$$
where $L(t)=\sum_{i=0}^{2g} a_it^i\in \mathbb{Z}[t]$ is a polynomial of degree $2g$. The polynomial $L(t)$ is called the $L$-polynomial of $F/\F_q$.

\begin{lemma}\label{lem:2.2}
Let $F/\F_q$ be an algebraic function field with genus $g$. Let $A_i$ be the number of all effective divisors of degree $i$.
Let $a_j$ be the coefficients of $L$-polynomial $L(t)=\sum_{j=0}^{2g} a_jt^j$. 
Then we have $$A_i=\sum_{j=0}^{\min\{i,2g\}} \frac{q^{i+1-j}-1}{q-1} a_j.$$
\end{lemma}
\begin{proof}
This result follows from the equation
\begin{align*}
\sum_{i=0}^\infty A_i t^i=Z(t)=\frac{L(t)}{(1-t)(1-qt)}&=\left(\sum_{j=0}^{2g} a_jt^j\right) \left(\sum_{k=0}^\infty t^k\right) \left(\sum_{u=0}^\infty q^u t^u\right)\\
&= \left(\sum_{j=0}^{2g} a_jt^j\right) \left(\sum_{k=0}^\infty \frac{q^{k+1}-1}{q-1}t^k\right).
\end{align*}
\end{proof}

\subsection{Codes}
Let $\F_q$ be the finite field with $q$ elements.
We denote a $q$-ary $(n,M,d)$ code as a code of length $n$, size $M$ and minimum distance $d$. 
There is a well-known upper bound on the size of codes which is called the Singleton bound \cite[Theorem 5.4.1]{LX04}.
\begin{lemma}\label{lem:2.3}
For any integer $q>1$, any positive integer $n$ and any integer $d$ with $1\le d\le n$, let $C$ be a $q$-ary $(n,M,d)$-code. Then we have
$M\le q^{n-d+1}.$
\end{lemma}

A linear code of length $n$ over $\F_q$ is a subspace of $\F_q^n$.
A linear code with length $n$, dimension $k$ and minimum distance $d$ is denoted as an $[n,k,d]$-linear code.
Any linear code achieving the Singleton bound, i.e., $k+d=n+1$, is called a maximum distance separable (MDS) code.

Denote by  $\Sigma$ the set $\F_q\cup \{\infty\}$. The size of $\Sigma$ is $|\Sigma|=q+1$. In this paper, we consider nonlinear codes over the alphabet $\Sigma$.
Let $\bx, \by$ be words of length $n$ over $\Sigma$. The Hamming distance of $\bx$ and $\by$, denoted by $d(\bx,\by)$, is defined to be the number of places at which $\bx$ and $\by$ differ. The minimum distance of $C$ is defined by
$d(C)=\min\{d(\bx,\by): \bx,\by\in C, \bx\neq \by\}.$
From Lemma \ref{lem:2.3}, any $(n, M, d)$-code over $\Sigma$ satisfies $M\le (q+1)^{n-d+1}$.
In order to obtain good lower bound on the size of code over $\Sigma$, one could make use of the following propagation rules given in Exercises of \cite[Chapter 6]{LX04}.

\begin{lemma}\label{lem:2.4}
\begin{itemize}
\item[(1)] (Alphabet extension) Let $s,r$ be two integers such that $s\ge r> 1$. We embed an alphabet $A$ of cardinality of $r$ into an alphabet $B$ of cardinality $s$.
Then any $(n,M,d)$-code $C$ over $A$ can be viewed as an $(n,M,d)$-code over $B$.

\item[(2)] (Alphabet restriction) Let $s,r$ be two integers such that $s\ge r> 1$. We embed an alphabet $A$ of cardinality of $r$ into $\mathbb{Z}_s$.
For an $(n,M,d)$-code $C$ over $\mathbb{Z}_s$, there exists an $r$-ary $(n,M^\prime, d^\prime)$-code with $M^\prime\ge M(r/s)^n$ and $d^\prime \ge d$.

\item[(3)] (Alphabet multiplication) Let $r$ and $s$ be two integers bigger than $1$. Let $C_1$ be an $(n,M_1,d_1)$-code over $\ZZ_r$, and let $C_2$ be an $(n,M_2,d_2)$-code over $\ZZ_s$. Then $C_1$ and $C_2$ can be viewed as codes over $\ZZ_{rs}$ by mapping $i(\text{mod } r)\in \ZZ_r$ and $i(\text{mod } s)\in \ZZ_s$ to $i(\text{mod } rs)\in \ZZ_{rs}$.
Furthermore, the code $$C_1+rC_2:=\{\bu+r\bv\in \ZZ_{rs}^n: \bu\in C_1,\bv\in C_2\}$$
is an $(n,M_1M_2,\min\{d_1,d_2\})$-code over $\ZZ_{rs}$.
\end{itemize}
\end{lemma}

\subsection{Algebraic geometry codes}
 Let $F/\F_q$ be an algebraic function field of genus $g$ with $N(F)$ rational places. Let $P_1,P_2,\cdots,P_n$ be rational points of $F$ and $D=\sum_{i=1}^nP_i$. For every divisor $G$ with $0< \deg(G)<n$ and $P_i\notin \supp(G)$, the algebraic geometry code $C(D,G)$ is defined as the image of evaluation map
$$\phi: \mL(G)\rightarrow \F_q^n,\  \phi(f)=(f(P_1),f(P_2),\cdots,f(P_n)).$$
From the Riemann-Roch Theorem, the dimension of $C(D,G)$ is $k=\ell(G)\ge \deg(G)-g+1$ and the minimum distance of $C(D,G)$ is lower bounded by $d\ge n-\deg(G)$.
It is easy to see that $n-g+1\le k+d\le n+1$. The following lemma gives an upper bound on the number of rational places of algebraic function fields over $\F_q$ from \cite[Theorem 5.2.3]{St09}.

\begin{lemma}\label{lem:2.6}
Let $F/\F_q$ be an algebraic function field of genus $g$ defined over the finite field $\F_q$ and let $N(F)$ be its number of rational places. Then we have
$$|N(F)-q-1|\le 2g\sqrt{q}.$$
\end{lemma}
The above bound given in Lemma \ref{lem:2.6} is called the Hasse-Weil bound. Any function field $F/\F_q$ of genus $g$ achieving the Hasse-Weil upper bound $q+1+2g\sqrt{q}$ is called maximal.
In order to construct good algebraic geometry codes, people need to use algebraic function fields with many rational places, especially maximal function fields \cite{BMXY13,BM18,GSX00,GK09,MX19}.

In particular, if $F/\F_q$ is an elliptic function field, then the elliptic code $C(D,G)$ is an $[n,k,d]$-linear code with $n\le k+d\le n+1$.
Hence, the elliptic code is an almost MDS code, i.e., $k+d=n$, or an MDS code.
Furthermore, the following result can be found from \cite[Proposition 3.4]{M93}.

\begin{lemma}\label{lem:2.5}
If a nontrivial elliptic MDS code has length $n>q+1$, then it is a $[6,3]$ code over $\F_4$ arising from a curve with $9$ rational points.
\end{lemma}

Let $N_q(g)$ be the maximum number of rational places of global function fields $F/\F_q$ of genus $g$.
A prime power $q=p^a$ is called exceptional if $a\ge 3$ is odd and $p$ divides $\lfloor 2\sqrt{q}\rfloor$.
From \cite[Corollary 9.94]{HKT08}, one has the following result.

\begin{lemma}\label{lem:2.7}
The value $N_q(1)$ can be determined explicitly as follows:
$$N_q(1)=\begin{cases} q+\lfloor 2\sqrt{q}\rfloor, & \text{ if } q \text{ is exceptional}, \\ q+1+\lfloor 2\sqrt{q}\rfloor, & \text{ otherwise}.\end{cases}$$
\end{lemma}

\section{A new construction of nonlinear codes}\label{sec:3}
Let $q$ be a prime power. Let $\F_q=\{\Ga_1, \Ga_2, \cdots, \Ga_q\}$ be the finite field with $q$ elements.
Denote by  $\Sigma$ the set $\F_q\cup \{\infty\}$. The size of $\Sigma$ is $|\Sigma|=q+1$.
In this section, we will propose a construction of $(q+1)$-ary nonlinear codes over the code alphabet $\Sigma$ via algebraic function fields by generalizing the ideas given in \cite{JMX21} and \cite{SX05}.

\begin{prop}\label{prop:3.1}
Let $F/\F_q$ be an algebraic function field with genus $g$ and $D$ be a divisor of $F$ with $\deg(D)=m\ge 2g-1$.
Let $Q_1,Q_2,\cdots,Q_t$ be distinct places of $F$ with $\deg(Q_i)=r_i$. 
Let $G=\sum_{i=1}^t m_iQ_i$  be a divisor of $F$ with $\deg(G)=\sum_{i=1}^t m_ir_i=s$ and $m_i\ge 1$ for $1\le i\le t$. Consider the set $$\mL_D(G)=\{f\in \mL(D+G)| \nu_{Q_i}(f)=-m_i-\nu_{Q_i}(D) \text{ for all } 1\le i\le t\}.$$
Then the cardinality of $\mL_D(G)$ is
$$|\mL_D(G)|=q^{m+s-g+1} \prod_{i=1}^t\left(1-\frac{1}{q^{r_i}}\right)\ge q^{m-g+1} (q-1)^s.$$
\end{prop}
\begin{proof}
From the definition of $\mL_D(G)$, it is clear that
$$\mL_D(G)=\mL(D+G)-\cup_{i=1}^t \mL(D+G-Q_i).$$
From the Riemann-Roch Theorem \cite[Theorem 1.5.14]{St09}, the size of  $\mL(D+G)$ is $$|\mL(D+G)|=q^{\deg(D+G)-g+1}=q^{m+s-g+1}.$$
From the inclusion-exclusion principle of combinatorics, the cardinality of $\cup_{i=1}^t \mL(D+G-Q_i)$ can be calculated explicitly as follows:
\begin{align*}
|\cup_{i=1}^t \mL(D+G-Q_i)| &= \sum_{k=1}^t(-1)^{k-1} \sum_{1\le i_1<i_2<\cdots < i_k\le t} |\cap_{j=1}^k \mL(D+G-Q_{i_j})|\\
&=\sum_{k=1}^t (-1)^{k-1} \sum_{1\le i_1<i_2<\cdots < i_k\le t}|\mL(D+G-\sum_{j=1}^k Q_{i_j})|\\
&=\sum_{k=1}^t (-1)^{k-1}  \sum_{1\le i_1<i_2<\cdots < i_k\le t} q^{m+s-g+1-\sum_{i=1}^k  r_{i_j}}\\
&= q^{m+s-g+1} \sum_{k=1}^t (-1)^{k-1}  \sum_{1\le i_1<i_2<\cdots < i_k\le t} q^{-\sum_{i=1}^k  r_{i_j}}\\
&= q^{m+s-g+1}  \left[1-\prod_{i=1}^t \left(1-\frac{1}{q^{r_i}}\right)\right].\\
\end{align*}
Hence, the cardinality of $\mL_D(G)$ is
\begin{align*}
|\mL_D(G)|&=q^{m+s-g+1}  \prod_{i=1}^t\left(1-\frac{1}{q^{r_i}}\right)\ge q^{m+s-g+1}  \prod_{i=1}^t\left(1-\frac{1}{q}\right)^{r_i}
\\&\ge q^{m+s-g+1}  \left(1-\frac{1}{q}\right)^{\sum_{i=1}^t m_ir_i}= q^{m+s-g+1} \left(1-\frac{1}{q}\right)^s=q^{m-g+1} (q-1)^s.
\end{align*}
\end{proof}

\begin{lemma}\label{lem:3.2}
Let $G_1$ and $G_2$ be two distinct positive divisors of $F$. Then we have $\mL_D(G_1)\cap \mL_D(G_2)=\emptyset.$
\end{lemma}
\begin{proof}
Suppose that there exists an element $f\in \mL_D(G_1)\cap \mL_D(G_2)$. We first claim that $\supp(G_1)=\supp(G_2)$. If there exists a place $Q\in \supp(G_i)\setminus \supp(G_j)$  for $i\neq j\in \{1,2\}$, then we have $\nu_Q(f)=-\nu_Q(G_i)-\nu_Q(D)\le -\nu_Q(D)-1$ and $\nu_Q(f)\ge -\nu_Q(G_j)-\nu_Q(D)=-\nu_Q(D)$. This is impossible.

If $f\in \mL_D(G_1)\cap \mL_D(G_2)$, then we have $\nu_{Q}(f)=-\nu_{Q}(G_1)-\nu_{Q}(D)=-\nu_{Q}(G_2)-\nu_{Q}(D)$
for any place $Q\in \supp(G_1)\cup \supp(G_2)$.
Hence, we have $\nu_{Q}(G_1)=\nu_{Q}(G_2)$ for any place $Q\in \mathbb{P}_F$, i.e., $G_1=G_2$, which is a contradiction to $G_1\neq G_2$.
\end{proof}

{\bf Construction:}  The construction of our nonlinear codes is given explicitly as follows.
Let $F/\F_q$ be an algebraic function field of genus $g$. Let $P_1,P_2,\cdots, P_n$ be rational places of $F/\F_q$.
Let $s$ be a positive integer.
For any positive integer $r\ge 4g+3$, there exist two places $R_{r+1}$ and $R_r$ in $\mathbb{P}_F$ with $\deg(R_{r+1})=r+1$ and $\deg(R_r)=r$ respectively from \cite[Corollary 5.2.10]{St09}.
Let $D=m(R_{r+1}-R_r)$ be a divisor of $F$ with $\deg(D)=m\ge 2g-1$.
Consider the set
$$\mL_s(D):= \bigcup_{G\ge 0, \deg(G)\le s} \mL_D(G),$$
where $G$ runs over all effective divisors of $F$ with $0\le \deg(G)\le s$. Here we assume that $\mL_D(0)=\mL(D)$.
Let $\Sigma$ be the set $\F_q\cup \{\infty\}$.
We define an evaluation map $\phi: \mL_s(D)\rightarrow \Sigma^n$ by putting $$\phi(f)=(f(P_1),f(P_2),\cdots, f(P_n))$$
for any element $f\in  \mL_s(D)$.
The image of $\phi$ together with $\{(\infty,\infty,\cdots,\infty)\}$ is our nonlinear code $C:=\phi(\mL_s(D))\cup \{(\infty,\infty,\cdots,\infty)\} \subseteq \Sigma^n.$

\begin{theorem}\label{thm:3.3}
Let $F/\F_q$ be an algebraic function field of genus $g$ with at least $n$ rational places, and let $A_i$ be the number of effective divisors of $F/\F_q$ with degree $i$.
Let $m\ge 2g-1$ and let $s$ be an non-negative integer with $n-m-2s>0$.
Then the code $C$ defined as above is a $(q+1)$-ary $(n,M,d)$-code with cardinality
$$M=|C|\ge 1+\sum_{i=0}^s (q-1)^i  q^{m-g+1}A_i,$$
and minimum distance
$$d\ge n-m-2s.$$
\end{theorem}
\begin{proof}
Under the assumption that the minimum distance of $C$ is $d\ge n-m-2s>0$, it is clear that the evaluation map $\phi$ is injective.
Hence, the cardinality of the code $C$ is lower bounded by
$$M=|C|\ge 1+\sum_{i=0}^s (q-1)^i  q^{m-g+1}A_i$$
from Proposition \ref{prop:3.1} and Lemma \ref{lem:3.2}.
It is easy to see that the Hamming distance of $\phi(f)$ and $(\infty,\infty,\cdots,\infty)$ is at least $ n-m-s$ for any $f\in \mL_s(D)$.
It will be sufficient to prove that the Hamming distance $d(\phi(f_1),\phi(f_2))$ of $\phi(f_1)$ and $\phi(f_2)$ is at least $n-m-2s$ for any two distinct elements $f_1,f_2\in \mL_s(D)$.

Assume that $f_1\in \mL_D(G_1)$ and $f_2\in \mL_D(G_2)$ for effective divisors $G_1,G_2$ with $\deg(G_1)\le s$ and $\deg(G_2)\le s$ respectively, then $f_1-f_2\in \mL(D+G_1\vee G_2)$.
If $P_i\in \supp(G_1)\cap \supp(G_2)$, then $f_1-f_2\in \mL(D+G_1+G_2-P_i)$ from Lemma \ref{lem:2.1}. Hence, we have $$f_1-f_2\in \mL\left(D+G_1+G_2-\sum_{P_i\in \supp(G_1)\cap \supp(G_2)} P_i\right).$$
Let $Z$ be a subset of $\{1,2,\cdots, n\}$ defined by
$$Z:=\{1\le j\le n| P_j\notin \supp(G_1)\cup \supp(G_2) \text{ and } f_1(P_j)=f_2(P_j)\}.$$
From Lemma \ref{lem:2.1}, we have $$0\neq f_1-f_2\in \mL\left(D+G_1+G_2-\sum_{P_i\in \supp(G_1)\cap \supp(G_2)} P_i-\sum_{j\in Z}P_j\right).$$ It follows that
$$m+\deg(G_1)+\deg(G_2)-|\supp(G_1)\cap \supp(G_2)|-|Z|\ge 0.$$
On the other hand, the Hamming distance of $\phi(f_1)$ and $\phi(f_2)$ is  $$d(\phi(f_1),\phi(f_2))\ge n-|\supp(G_1)\cap \supp(G_2)|-|Z|.$$
Hence, the minimum distance $d$ of the code $C$ is lower bounded by
\begin{align*} d &\ge n-|\supp(G_1)\cap \supp(G_2)|-|Z|\\ & \ge n-m-\deg(G_1)-\deg(G_2)\\& \ge n-m-2s.\end{align*}
\end{proof}

\begin{cor}\label{cor:3.4}
Let $F/\F_q$ be an algebraic function field of genus $g$ with at least $n$ rational places, and let $A_i$ be the number of effective divisors of $F/\F_q$ with degree $i$.
Let $m$ be a positive integer with $m\ge 2g-1$.
For a fixed minimum distance $2\le d\le n-m$, there exists a $(q+1)$-ary $(n,M,d)$-code with cardinality
$$M\ge 1+\max_{2g-1\le m\le n-d} \left\{ \sum_{i=0}^{[(n-d-m)/2]} (q-1)^i  q^{m-g+1}A_i\right\},$$
here $[x]$ is the integer part of $x\in \mathbb{R}$.
\end{cor}
\begin{proof}
This corollary follows from Theorem \ref{thm:3.3} and \cite[Theorem 6.1.1]{LX04}.
\end{proof}

\section{Nonlinear codes via elliptic curves}\label{sec:4}
In this section, we provide an explicit construction of nonlinear codes via elliptic curves given in Section \ref{sec:3}.

Let $E/\F_q$ be an elliptic curve defined over a finite field $\F_q$. Let $N(E)$ be the number of rational points of elliptic curve $E/\F_q$.
From \cite[Theorem 5.1.15]{St09}, the $L$-polynomial of the elliptic curve $E/\F_q$ is given by $L(t)=1+(N(E)-q-1)t+qt^2\in \mathbb{Z}[t]$, i.e., $a_0=1$, $a_1=N(E)-q-1$, $a_2=q$ and $a_j=0$ for $j\ge 3$. From Lemma \ref{lem:2.2}, the number of effective divisors of $E/\F_q$ with degree $i$ is given by $A_i=\sum_{j=0}^{i}a_j (q^{i+1-j}-1)/(q-1).$
Let $m\ge 2g(E)-1=1$ and $s$ be two non-negative integers.
From Theorem \ref{thm:3.3}, there exists a $(q+1)$-ary $(n,M,d)$ nonlinear code with length $n=N(E)$, size
$M\ge 1+\sum_{i=0}^s \left(q-1\right)^i  q^{m}A_i,$
and minimum distance $d\ge n-m-2s>0.$ Hence, the following proposition follows from Corollary \ref{cor:3.4}.

\begin{prop}\label{prop:4.1}
Let $E/\F_q$ be an elliptic curve with $N(E)$ rational points.
For $q+3\le n \le N(E)$ and $2\le d\le n-1$, there exists a $(q+1)$-ary $(n,M,d)$-nonlinear code $C_E$ with cardinality
$M=|C_E|\ge 1+ \sum_{i=0}^{[(n-d-m)/2]} (q-1)^i  q^{m}A_i$
 for all $1\le m\le n-d$. 
\end{prop}

In the following, we want to compare our nonlinear codes via elliptic curves with the codes obtained from propagation rules given in Lemma \ref{lem:2.4}.

\subsection{Alphabet extension}
In this subsection, we compare our nonlinear codes via elliptic curves with the codes constructed via the alphabet extension of elliptic codes.
If $q+3\le n\le N(E)$, then there exists a $q$-ary $[n,n-d,d]$-linear code constructed from elliptic curve $E/\F_q$.
Furthermore, the nontrivial $q$-ary $[n,n-d+1,d]$-elliptic MDS code doesn't exist from Lemma \ref{lem:2.5}, i.e., the $q$-ary $[n,n-d,d]$-linear code is the best-known linear code for given length $n$ and minimum distance $d$ in the literature.
From Lemma \ref{lem:2.4}, there exists a $(q+1)$-ary $(n,q^{n-d},d)$-nonlinear code via code alphabet extension.

\begin{prop}\label{prop:4.2}
Let $E/\F_q$ be an elliptic curve with $N(E)$ rational points. For $q+3\le n \le N(E)$ and $2\le d\le n-1$, there exists a $(q+1)$-ary $(n,M,d)$-nonlinear code $C_E$ with cardinality bigger than
$q^{n-d},$
i.e., the size of the $(q+1)$-ary nonlinear code $C_E$ via elliptic curves is larger than the size of codes constructed from code alphabet extension of elliptic codes.
\end{prop}
\begin{proof}
From Proposition \ref{prop:4.1}, for $m=n-d$, there exists a $(q+1)$-ary $(n,M,d)$-nonlinear code $C_E$ with cardinality
$$M=|C_E|\ge 1+  (q-1)^0 q^{m}A_0=1+q^{n-d}>q^{n-d}.$$
\end{proof}

\subsection{Alphabet restriction}
If $q+2$ is a prime power as well, then there exists a $(q+2)$-ary $[n,n-d,d]$-linear code from elliptic codes.
From Lemma \ref{lem:2.4}, there exists a $(q+1)$-ary $\left(n,M^\prime\ge \frac{(q+1)^n}{(q+2)^d}, d\right)$-nonlinear code via code alphabet restriction of elliptic codes. In the case where $q+2$ is not a prime, we are not sure if there are still exists a $(q+2)$-ary $(n,(q+2)^{n-d},d)$-code for $n>q+1$. Nevertheless, no matter whether $q+2$ is a prime or not, we use $(q+2)$-ary $(n,(q+2)^{n-d},d)$-codes to compare with our codes in the following proposition.

\begin{prop}\label{prop:4.3}
Let $E/\F_q$ be an elliptic curve with $N(E)$ rational points. If  $q+1\le n \le N(E)$ and $d\ge n\cdot \ln(1+\frac{1}{q})/\ln(1+\frac{2}{q})$, then there exists a $(q+1)$-ary $(n,M,d)$ nonlinear code $C_E$ with cardinality larger than $\frac{(q+1)^n}{(q+2)^d},$
 i.e., the size of the $(q+1)$-ary nonlinear code $C_E$ is larger than the size of codes constructed from code alphabet restriction of $(q+2)$-ary $(n,(q+2)^{n-d},d)$-codes for sufficiently large $d$.
\end{prop}
\begin{proof}
From Proposition \ref{prop:4.1}, 
there exists a $(q+1)$-ary $(n,M,d)$ nonlinear code $C_E$ with cardinality
$$M=|C_E|\ge 1+\sum_{i=0}^1 (q-1)^i  q^{n-d-2}A_i\ge 1+q^{n-d-2}[1+(q-1)n].$$ It is easy to verify that
$$q^{n-d-2}[1+(q-1)n]\ge \frac{(q+1)^n}{(q+2)^d} \quad \Leftrightarrow  \quad \frac{(q+2)^d}{q^d}\ge \frac{(q+1)^n}{q^n}\cdot \frac{q^2}{1+(q-1)n}.$$
If $n\ge q+1$ and $d\ge n\cdot \ln(1+\frac{1}{q})/\ln(1+\frac{2}{q})$, then we have
$M> (q+1)^{n}/(q+2)^{d}.$
\end{proof}

\begin{rmk}
If $q$ is a prime power, then $q+2$ may not be a prime power. Let $n$ be a positive integer with $q+1\le n\le N_q(1)$.
Let $q+a$ be the least prime power satisfying $q+a\ge n-1$.
Then there exists a $(q+a)$-ary $[n,n+1-d,d]$-MDS code from rational algebraic geometry codes.
Hence, we can obtain a $(q+1)$-ary $\left (n, M^\prime\ge \frac{(q+1)^n}{(q+a)^{d-1}},d\right)$-nonlinear code via code alphabet restriction of the above MDS code.
In particular, if $n=q+1+\lfloor 2\sqrt{q}\rfloor$, then we choose $a\ge \lfloor 2\sqrt{q}\rfloor$ to be an integer such that $q+a$ is a prime power.
From Proposition \ref{prop:4.3},  there exists a $(q+1)$-ary $(n,M,d)$-nonlinear code $C_E$ with cardinality $M\ge 1+q^{n-d-2}[1+(q-1)n].$
It is easy to verify that
$$q^{n-d-2}[1+(q-1)n]\ge (q+a)^{n-d+1}\left(\frac{q+1}{q+a}\right)^n=\frac{(q+1)^n}{(q+a)^{d-1}}$$
if and only if
$$\frac{(q+a)^d}{q^d}\ge \frac{(q+1)^n(q+a)}{q^n}\cdot \frac{q^2}{1+(q-1)n}.$$
If $d\ge [n\cdot \ln(1+\frac{1}{q})+\ln(q+a)]/\ln(1+\frac{a}{q})$, then we have 
$$M> \frac{(q+1)^n}{(q+a)^{d-1}}.$$
\end{rmk}

\subsection{Numerical examples}
In this subsection, we provide numerical examples from nonlinear codes via elliptic curves and compare our nonlinear codes with other $(q+1)$-ary nonlinear codes via code alphabet extension and restriction. Although Proposition \ref{prop:4.3} shows that our codes are better than those obtained from alphabet restriction for sufficiently large minimum distance $d$, our numerical results show that even for small $d$, our codes still outperform those obtained from code alphabet restriction.

\begin{ex}\label{ex:4.4}
{\rm Let $E/\F_5$ be the function field defined by $E=\F_5(x,y)$ with $y^2=3(x^4+2)$
given in \cite{NX97}.
All rational places of $\F_5(x)$ except the infinite place $\infty$ split completely in $E/\F_5(x)$, and the genus of $E$ is one from the theory of Kummer extension \cite[Proposition 3.7.3]{St09}. Hence, the elliptic function field $E$ has $10$ rational places which achieves the Serre bound, i.e., $N_5(1)=10$.
The $L$-polynomial of $E/\F_q$ is given by $L(t)=1+4t+5t^2\in \mathbb{Z}[t]$, i.e.,  $a_0=1$, $a_1=4$, $a_2=5$ and $a_j=0$ for $j\ge 3$. From Lemma \ref{lem:2.2}, the number of effective divisors of $E/\F_5$ of degree $i$ is given by $A_i=\sum_{j=0}^i(5^{i+1-j}-1)a_j/4.$
Let $d$ be an integer with $2\le d\le 9$. From Theorem \ref{thm:3.3} and Corollary \ref{cor:3.4}, there exists a $6$-ary $(10,M,d)$-nonlinear code with size
$$M\ge 1+\sum_{i=0}^{[(10-d-m)/2]} 4^i \cdot 5^{m}A_i,$$
for any integer $1\le m\le 10-d$. }
\end{ex}

In the following table, we compare the codes given in Example \ref{ex:4.4} with those obtained via code  alphabet extension and restriction. Note that to obtain codes via code extension and restriction, we have to start with a code of the best-known parameters. However, we are lack of nonlinear codes with the best-known parameters. Instead, we choose linear codes with the best-known parameters given in the online table \cite{G22}.

We use the case where $q=5,n=10$ and $d=4$ to illustrate the following table. In this case, we start with  $5$-ary  $[10,6,4]$ and $7$-ary $[10,6,4]$-linear codes  and then apply code alphabet extension and restriction to obtain $6$-ary codes with sizes $15625$ and $25184$, respectively.
From the online table \cite{G22}, there exist $2$-ary $[10,5,4]$ and $3$-ary $[10,6,4]$-linear codes and then apply code alphabet multiplication given in Lemma \ref{lem:2.4} to obtain a $6$-ary code with size $23328$.
In the last column, we provide code sizes obtained from Example \ref{ex:4.4}.

{\footnotesize
\begin{center}~\label{table:1}
Table I\\   Comparison of sizes of $6$-ary codes of length $10$\\ \smallskip
{\rm
\begin{tabular}{|c|c|c|c|c|}\hline\hline

Distance $d$ & Alphabet extension & Alphabet restriction & Alphabet multiplication & \multicolumn{1}{c|}{Example \ref{ex:4.4}}\\ \hline

 {$4$} &$15625$ &  $25184$ & $23328$ & \multicolumn{1}{c|}{\bf 25626}\\\cline{1-5}

{$5$} & $3125$ &$3598$ & $1000$ &  \multicolumn{1}{c|}{\bf 5126}\\  \cline{1-5}

{$6$} & $625$& $514$ & $324$ & \multicolumn{1}{c|}{\bf 1026}\\  \cline{1-5}

{$7$} & $125$& $74$ & $18$ & \multicolumn{1}{c|}{\bf 206}\\  \cline{1-5}

{$8$} & $25$& $11$  & $6$ & \multicolumn{1}{c|}{\bf 42}\\  \hline\hline
\end{tabular}}
\end{center}}

\begin{ex}\label{ex:4.5}
{\rm Let $E/\F_9$ be the function field defined by $E=\F_9(x,y)$ with $y^2=x^4+1$ given in \cite{NX97}.
In fact, $E/\F_9$ is a maximal elliptic function field with $16$ rational places from \cite{NX97} and the $L$-polynomial of $E/\F_q$ is $L(t)=1+6t+9t^2\in \mathbb{Z}[t]$, i.e., $a_0=1$, $a_1=6$, $a_2=9$ and $a_j=0$ for $j\ge 3$. From Lemma \ref{lem:2.2}, we have
$A_i=\sum_{j=0}^i(9^{i+1-j}-1) a_j/8.$
Let $d$ be an integer with $2\le d\le 15$.
From Theorem \ref{thm:3.3} and Corollary \ref{cor:3.4}, there exists a $10$-ary $(16,M,d)$ nonlinear code with size
$$M\ge 1+\sum_{i=0}^{[(16-d-m)/2]} 8^i \cdot 9^{m}A_i,$$
for any integer $1\le m\le 16-d$. }
\end{ex}

 From Example \ref{ex:4.4}, the size of codes via code alphabet multiplication turns out to be not good enough for large minimum distance $d$.
Hence, we only compare the codes given in Example \ref{ex:4.5} with those obtained via code  alphabet extension and restriction in the following table.
 In particular,  we use $11$-ary $[16, k, 16-k]$-linear codes for comparison with the codes via code alphabet restriction.

{\footnotesize
\begin{center}~\label{table:2}
Table II\\ Comparison of sizes of $10$-ary codes of length $16$ \\ \smallskip
{\rm
\begin{tabular}{|c|c|c|cl}\hline\hline
Minimum distance $d$ & Alphabet extension & Alphabet restriction & \multicolumn{1}{c|}{Example \ref{ex:4.5}}\\ \hline

{$7$} & $387,420,489$& $513,158,119$ & \multicolumn{1}{c|}{\bf 617,003,002}\\  \cline{1-4}

{$8$} & $43,046,721$& $46,650,739$ & \multicolumn{1}{c|}{\bf 68,555,890}\\  \cline{1-4}

{$9$} & $4,782,969$& $4,240,977$ & \multicolumn{1}{c|}{\bf 7,617,322}\\  \cline{1-4}

{$10$} & {$531,441$} &$385,544$ & \multicolumn{1}{c|}{\bf 846,370}\\  \cline{1-4}

{$11$} & $59,049$& $35,050$ & \multicolumn{1}{c|}{\bf 94,042}\\  \cline{1-4}

{$12$} & $6,561$& $3,187$ & \multicolumn{1}{c|}{\bf 10,450}\\  \cline{1-4}

{$13$} & $729$& {$290$}  & \multicolumn{1}{c|}{\bf 1,162}\\  \cline{1-4}

{$14$} & $81$& {$27$}  & \multicolumn{1}{c|}{\bf 130}\\  \hline\hline
\end{tabular}
}
\end{center}
}

\section{Nonlinear codes via maximal function fields}\label{sec:5}
In this section, we provide an explicit construction of nonlinear codes via maximal function fields given in Section \ref{sec:3}.

Let $F/\F_{q}$ be a maximal function field of genus $g$.  If $g\ge 1$, then $q$ must be a square of a prime power. Otherwise, $F/\F_q$ is the rational function field over $\F_q$ for any prime power.
The number of rational places of $F$ is $N(F)=q+1+2g\sqrt{q}$ and  the $L$-polynomial of $F/\F_{q}$ is $L(t)=(1+\sqrt{q}t)^{2g}\in \mathbb{Z}[t]$.
Hence, we have $a_j=\binom{2g}{j} \sqrt{q}^j$ for $0\le j\le 2g$ and $a_j=0$ for $j\ge 2g+1$.
From Lemma \ref{lem:2.2}, the number of effective divisors of $F/\F_q$ is $A_i=\sum_{j=0}^{i}a_j (q^{i+1-j}-1)/(q-1).$
Let $m\ge 2g-1$ and let $s$ be an non-negative integer with $n-m-2s>0$.
From Theorem \ref{thm:3.3}, there exists a $(q+1)$-ary $(n,M,d)$-nonlinear code with length $q+1\le n\le q+1+2g\sqrt{q}$, size $M\ge 1+\sum_{i=0}^s (q-1)^i  q^{m+1-g}A_i,$ and minimum distance $d\ge n-m-2s.$
From Corollary \ref{cor:3.4}, there exists a $(q+1)$-ary $(n,M,d)$-nonlinear code $C$ with size $$M=|C|\ge 1+\max_{2g-1\le m\le n-d} \left\{ \sum_{i=0}^{[(n-d-m)/2]} (q-1)^i  q^{m-g+1}A_i\right\}.$$

\subsection{Alphabet extension}
If $q+1\le n\le q+1+2g\sqrt{q}$, then there exists a $q$-ary $[n,n-g+1-d,d]$-linear code constructed from the maximal function field $F/\F_q$.
From Lemma \ref{lem:2.4}, there exists a $(q+1)$-ary $(n,q^{n-g+1-d},d)$-nonlinear code via code alphabet extension of algebraic geometry codes.

\begin{prop}\label{prop:5.1}
Let $F/\F_q$ be a maximal function field with genus $g$. For $q+1\le n \le q+1+2g\sqrt{q}$ and $2\le d\le n-g$, there exists a $(q+1)$-ary $(n,M,d)$-nonlinear code $C_F$ with cardinality larger than
$q^{n-g+1-d},$
i.e., the size of the $(q+1)$-ary nonlinear code $C_F$ is larger than the size of codes constructed from code alphabet extension of $[n,n-g+1-d,d]$-algebraic geometry codes.
\end{prop}
\begin{proof}
Let $m=n-d-2$. From Theorem \ref{thm:3.3} and the fact that the number of effective divisors of $E$ degree one is $A_1=q+1+2g\sqrt{q}$,  there exists a $(q+1)$-ary $(n,M,d)$-nonlinear code $C_F$ with cardinality
 $$M\ge 1+\sum_{i=0}^1 (q-1)^i  q^{n-d-2-g+1}A_i= 1+q^{n-g-d-1}[1+(q-1)(q+1+2g\sqrt{q})].$$
It is easy to verify that
$$M\ge 1+q^{n-g-d-1}[1+(q-1)n]\ge 1+q^{n-g-d-1}[1+(q-1)(q+1)]>q^{n-g+1-d}.$$
\end{proof}

\subsection{Alphabet restriction}
If $q+2$ is a prime power as well, then there exists a $(q+2)$-ary $[n,n-g+1-d,d]$-linear code from algebraic geometry codes.
Since there may be a lack of the parameters of the optimal linear codes for given $q, n$ and $d$, the algebraic geometry codes are good candidate for optimal linear codes for large length $n$ compared with $q$.
From Lemma \ref{lem:2.4}, there exists a $(q+1)$-ary $\left(n,M^\prime\ge \frac{(q+1)^n}{(q+2)^{d+g-1}}, d\right)$-nonlinear code via code alphabet restriction of algebraic geometry codes. Again, in the case where $q+2$ is not a prime, we are not sure if there still exists a $(q+2)$-ary $(n,(q+2)^{n-d-g+1},d)$-code for $n=q+1+2g\sqrt{q}$. Nevertheless, no matter whether $q+2$ is a prime or not, we use $(q+2)$-ary $(n,(q+2)^{n-d-g+1},d)$-codes to compare with our codes in the following proposition.

\begin{prop}\label{prop:5.2}
Let $F/\F_q$ be a maximal function field with genus $g$.
If $q+1\le n \le q+1+2g\sqrt{q}$ and $d\ge 1-g+n\cdot \ln(1+\frac{1}{q})/\ln(1+\frac{2}{q})$, then there exists a $(q+1)$-ary $(n,M,d)$-nonlinear code $C_F$ with cardinality larger than
$\frac{(q+1)^n}{(q+2)^{d+g-1}},$
 i.e.,
the size of the $(q+1)$-ary nonlinear code $C_F$ is larger than the one constructed from code alphabet restriction of $(n,(q+2)^{n-d-g+1},d)$-codes for sufficiently large $d$.
\end{prop}
\begin{proof}
From Proposition \ref{prop:5.1},  there exists a $(q+1)$-ary $(n,M,d)$-nonlinear code $C_F$ with cardinality
 $M=|C_F|\ge1+q^{n-g-d-1}[1+(q-1)n].$
 It is easy to verify that
$$q^{n-g-d-1}[1+(q-1)n]\ge \frac{(q+1)^n}{(q+2)^{d+g-1}}$$
if and only if
$$\frac{(q+2)^{d+g-1}}{q^{d+g-1}}\ge \frac{(q+1)^n}{q^n}\cdot \frac{q^2}{1+(q-1)n}.$$
If $d\ge 1-g+n\cdot \ln(1+\frac{1}{q})/\ln(1+\frac{2}{q})$, then we have
$$M> (q+1)^{n}/(q+2)^{d+g-1}.$$
This completes the proof.
\end{proof}

\subsection{Numerical examples}
In this subsection, we provide numerical examples from our nonlinear codes via maximal function fields and compare our nonlinear codes with other $(q+1)$-ary nonlinear codes via code alphabet extension and restriction.

\begin{ex}
{\rm Let $F/\F_q$ be the rational function field $\F_q(x)$. Its $L$-polynomial is $L(t)=1\in \mathbb{Z}[t]$.
From Lemma \ref{lem:2.2}, we have $A_i=(q^{i+1}-1)/(q-1)$ for all $i\in \mathbb{N}$.
From Theorem \ref{thm:3.3}, there exists a $(q+1)$-ary $(n,M,d)$-nonlinear code with length $n=q+1$ and size
$$M\ge  1+q^{n+1-d}> (q+1)^{n}/(q+2)^{d-1}. $$
From Propositions \ref{prop:5.1} and \ref{prop:5.2}, the size of our nonlinear codes via the rational function field is better than the one obtained from code alphabet extension and restriction of MDS codes.
In particular, if $D=0$, then the nonlinear code constructed in Theorem \ref{thm:3.3} is the same as the one given in \cite{JMX21}.
The size of such code $C$ has been determined explicitly as $|C|=q^{2s+1}+q^{2s}-2q^s+2$ and the minimum distance of $C$ is exactly $d=q+1-2s$ from \cite[Theorem III.5]{JMX21}.
Furthermore, it has been shown that $q^{2s+1}+q^{2s}-2q^s+2>(q+1)^{2s}$ for $s\le q/2$. Hence, the code $C$ is a $(q+1)$-ary $(q+1,M,d)$-nonlinear code satisfying
$n<\log_{q+1} M+d\le n+1.$
It turns out that the code $C$ is quite good at the trade-off between information rate and minimum distance.}
\end{ex}

\begin{ex}\label{ex:5.4}
{\rm Let $H/\F_{9}$ be the Hermitian function field $H=\F_9(x,y)$ defined by $y^3+y=x^4$.
From \cite[Lemma 6.4.4]{St09}, $H$ is a maximal function field of genus $g=3$ and the number of rational places of $H$ is $N(H)=28$.
Hence, the $L$-polynomial of $H/\F_{9}$ is given by $L_H(t)=(1+3t)^{6}\in \mathbb{Z}[t]$, i.e., $a_j=\binom{6}{j} \cdot 3^j$ for $0\le j\le 6$ and $a_j=0$ for $j\ge 7$.
From Lemma \ref{lem:2.2}, the number of effective divisors of $H/\F_9$ is
$A_i=\sum_{j=0}^{i}a_j (9^{i+1-j}-1)/8.$
For $2\le d\le 23$, from Theorem \ref{thm:3.3} and Corollary \ref{cor:3.4}, there exists a $10$-ary $(28,M,d)$-nonlinear code with size
$M\ge 1+\sum_{i=0}^{[(28-d-m)/2]} 8^i \cdot 9^{m-2}A_i,$
for any $5\le m \le 28-d.$}
\end{ex}
{\footnotesize

\begin{center}~\label{table:3}

Table III\\ Comparison of sizes of $10$-ary codes of length $28$ \\ \smallskip

{\rm

\begin{tabular}{|c|c|c|c|}\hline\hline

Minimum distance $d$ & Alphabet extension code size& Alphabet restriction code size & {Example \ref{ex:5.4}}\\ \hline
{$6$} & {$1.22\times 10^{19}$} &$4.67\times 10^{19}$ & ${\bf 4.85\times10^{19}}$\\   \hline

{$7$} & $1.35\times 10^{18}$& $4.24\times 10^{18}$ & ${\bf 5.39\times 10^{18}}$\\  \hline

{$8$} & $1.50\times 10^{17}$& $3.86\times 10^{17}$ & ${\bf 5.99\times 10^{17}}$\\   \hline

{$9$} & $1.67\times 10^{16}$& {$3.50\times 10^{16}$}  & ${\bf 6.66\times 10^{16}}$\\   \hline

{$10$} & {$1.85\times 10^{15}$} &$3.19\times 10^{15}$ &${\bf 7.40\times 10^{15}}$\\  \hline

{$11$} & $2.06\times 10^{14}$& $2.90\times 10^{14}$ & ${\bf 8.22\times 10^{14}}$\\  \hline

{$12$} & $2.29\times 10^{13}$& $2.63\times 10^{13}$ & ${\bf 9.13\times 10^{13}}$\\  \hline

{$13$} & $2.54\times 10^{12}$& {$2.39\times 10^{12}$}  & ${\bf 1.01\times 10^{13}}$\\  \hline

{$14$} & $2.82\times 10^{11}$& {$2.18\times 10^{11}$}  & ${\bf 1.12\times 10^{12}}$\\  \hline
{$15$} & $3.14\times 10^{10}$& {$1.98\times 10^{10}$}  & ${\bf 1.25\times 10^{11}}$\\  \hline

{$16$} & $3.49\times 10^{9}$& {$1.80\times 10^{9}$}  & ${\bf 1.39\times 10^{10}}$\\  \hline

{$17$} & $3.87\times 10^{8}$& {$1.64\times 10^{8}$}  & ${\bf 1.54\times 10^{9}}$\\  \hline

{$18$} & $4.30\times 10^{7}$& {$1.49\times 10^{7}$}  & ${\bf 1.71\times 10^{8}}$\\  \hline

{$19$} & $4.78\times 10^{6}$& {$1.35\times 10^{6}$}  & ${\bf 1.91\times 10^{7}}$\\  \hline

{$20$} & $5.31\times 10^{5}$& {$1.23\times 10^{5}$}  &$ {\bf 2.12\times 10^{6}}$\\  \hline

{$21$} & $59,049$& {$11,168$}  & ${\bf 235,882}$\\  \hline

{$22$} & $6,561$& {$1,016$}  & ${\bf 26,210}$\\  \hline \hline

\end{tabular}

}

\end{center}

}
Note that for those comparison, we are lack of the parameters of $11$-ary codes from the online table \cite{G22}  for code alphabet restriction, here we use $11$-ary $[28,26-d,d]$-algebraic geometry codes.

\end{document}